\documentclass[a4paper,aps,prl,showpacs,twocolumn,superscriptaddress]{revtex4-1}   

\usepackage[utf8]{inputenc}  
\usepackage[T1]{fontenc}     
\usepackage[british]{babel}  
\usepackage[osf,sc]{mathpazo}
\usepackage[usenames,dvipsnames]{color} 
\usepackage[colorlinks,citecolor=blue,linkcolor=magenta,urlcolor=blue]{hyperref}  
\usepackage{graphicx} 
\usepackage[caption=false]{subfig}
\usepackage{tikz}
\usepackage[babel]{microtype}  
\usepackage{amsmath,amssymb,amsthm,bm,mathtools,amsfonts,mathrsfs,bbm,dsfont} 
\usepackage{xspace}  
\usepackage{multirow}
\usepackage{verbatim}
\usepackage{physics}
\usepackage{xcolor}

\newcommand{\id}{\ensuremath{\mathds{1}}}

\usepackage{bbold}

\newtheoremstyle{mystyle}
  {6pt}
  {6pt}
  {\normalfont}
  {0pt}
  {\bf}
  {.}
  { }
  {}

\theoremstyle{mystyle}
\newtheorem{observation}{Observation}
\newtheorem{lemma}[observation]{Lemma}
\newtheorem{theorem}[observation]{Theorem}
\newtheorem{definition}[observation]{Definition}

\begin{document}

\title{Characterizing Quantum Networks: Insights from Coherence Theory} 

\author{Tristan Kraft}
\affiliation{Naturwissenschaftlich-Technische Fakultät, Universität Siegen, Walter-Flex-Straße 3, 57068 Siegen, Germany}
\author{Cornelia Spee}
\affiliation{Institute for Quantum Optics and Quantum Information, Austrian Academy of Sciences, 1090 Vienna, Austria}
\author{Xiao-Dong Yu}
\affiliation{Naturwissenschaftlich-Technische Fakultät, Universität Siegen, Walter-Flex-Straße 3, 57068 Siegen, Germany}
\author{Otfried G\"uhne}
\affiliation{Naturwissenschaftlich-Technische Fakultät, Universität Siegen, Walter-Flex-Straße 3, 57068 Siegen, Germany}

\date{\today}

\begin{abstract}
Networks based on entangled quantum systems enable interesting applications 
in quantum information processing and the understanding of the resulting
quantum correlations is essential for advancing the technology. We show
that the theory of quantum coherence provides powerful tools for analyzing
this problem. For that, we demonstrate that a recently proposed approach 
to network correlations based on covariance matrices can be improved and
analytically evaluated for the most important cases.
\end{abstract}

\maketitle

{\it Introduction.---}
Quantum networks~\cite{Kimble2008,Sangouard2011,Simon2017,Wehner2018} 
have recently attracted much interest as they have been identified as 
a promising platform for quantum information processing, such as 
long-distance quantum communication~\cite{Cirac1998,Gisin2002Review}. 
In an abstract sense, a quantum network consists of several sources, 
which distribute entangled quantum states to spatially separated nodes,
then the quantum information is processed locally in these nodes. This 
may be seen as a generalization of a classical causal model~\cite{Spirtes2000, Pearl2009}, 
where the shared classical information between the nodes is replaced by 
quantum states. Clearly, it is important to understand the quantum correlations 
that arise in such a quantum network. Recent developments have shown that 
the network structure and topology leads to novel notions of nonlocality~\cite{Renou2019triangle,Gisin2020}, as well as new concepts of 
entanglement and separability \cite{navascues2020,kraft2020,luo2020},
which differ from the traditional concepts and definitions~\cite{Acin2001, GuehneToth2009}. Dealing with these new concepts requires theoretical 
tools for their analysis. So far, examples of entanglement criteria for the 
network scenario have been derived using the mutual information \cite{navascues2020,kraft2020}, the fidelity with pure states 
\cite{kraft2020,luo2020}, or covariance matrices build from measurement
probabilities \cite{kela2020, aberg2020}, but these ideas work either only
for specific examples, or require numerical optimizations 
for their evaluation.

In this paper we demonstrate that the theory of quantum coherence 
provides powerful tools for analyzing correlations in quantum networks.
In recent years, quantum coherence was under intense research,
it was demonstrated that coherence is essential in quantum information 
applications and entanglement generation~\cite{Aberg2006,Baumgratz2014,Killoran2016,Regula_2018} and a resource
theory of it has been developed~\cite{StreltsovReview,ChitambarReview,Winter2016,Chitambar2016,Streltsov2017}.
We provide a direct link between the theory of multisubspace coherence~\cite{ringbauer2018,kraft2019} and the approach to quantum networks using covariance
matrices established in Refs.~\cite{kela2020, aberg2020}. This 
allows to solve analytically the criteria developed there for important
cases; furthermore, some conjectures can be proved and, besides that, our methods can be applied to large networks for which tools based on numerical optimization are infeasible. We note that, since the covariance matrix approach is essentially
a tool coming from classical causal models~\cite{kela2020}, our results demonstrate
that results from the theory of {\it quantum} coherence are useful beyond the level of quantum states
for the analysis of {\it classical} networks.

\begin{figure}[t!]
\centering
\subfloat[\label{sfig:trianglenetwork}]{
  \includegraphics[width=.4\linewidth]{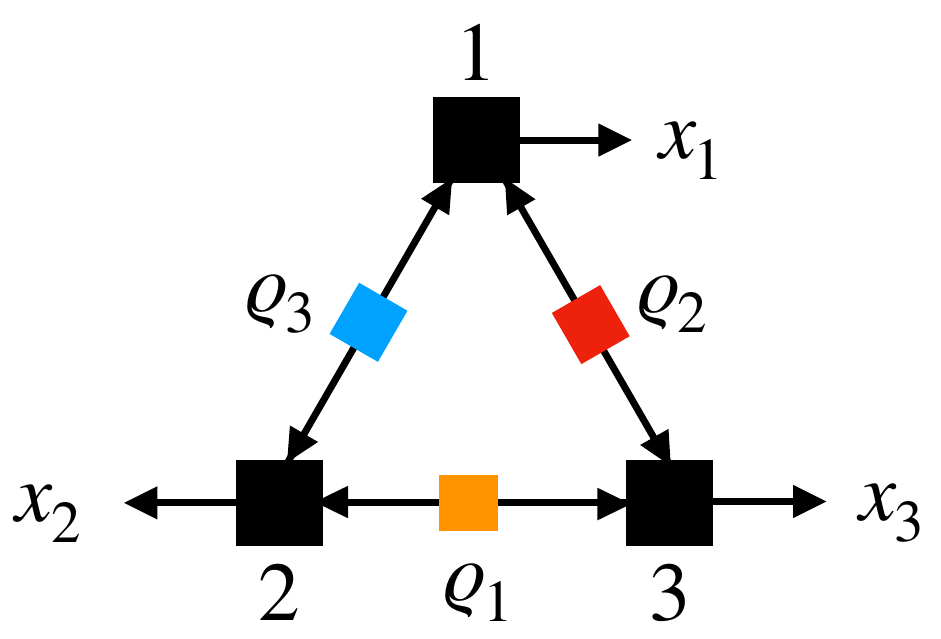}
}
\subfloat[\label{sfig:4completeNetwork}]{
  \includegraphics[width=.36\linewidth]{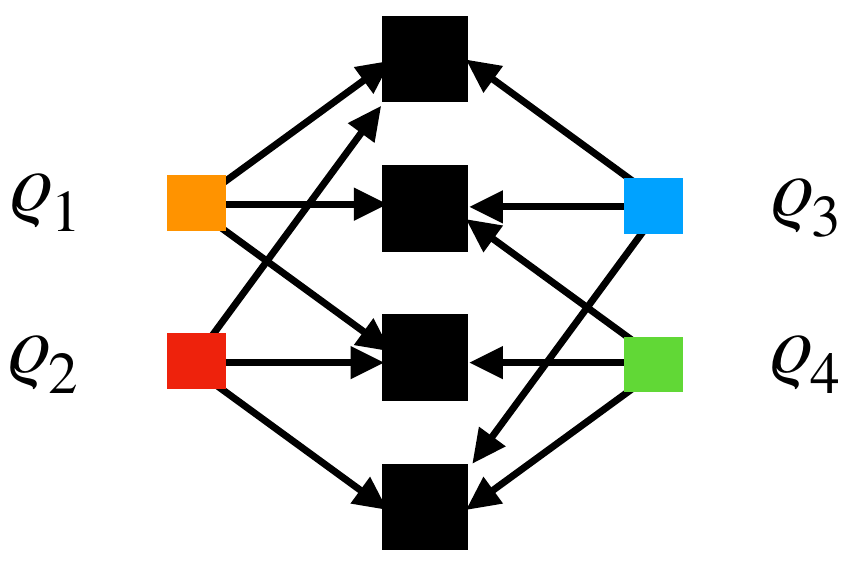}
}
\hspace{0.3cm}
\subfloat[\label{sfig:decomposition}]{
  $\displaystyle \Gamma(\mathbf{v})=\mqty(\color{cyan}\blacksquare & \color{cyan}\blacksquare & \phantom{\blacksquare} \\ \color{cyan}\blacksquare & \color{cyan}\blacksquare &  \\  &  & )_{3} 
  + \mqty(\color{red}\blacksquare & \phantom{\blacksquare} & \color{red}\blacksquare \\  &  &  \\ \color{red}\blacksquare &  & \color{red}\blacksquare )_{2} 
  + \mqty(\phantom{\blacksquare} &  &   \\ & \color{orange}\blacksquare & \color{orange}\blacksquare \\  & \color{orange}\blacksquare  & \color{orange}\blacksquare )_{1} $
}
\caption{(a) In its simplest scenario, the triangle network consists of three nodes that produce measurement outcomes $x_1,\dots,x_3$ and three sources that distribute bipartite entanglement that is shared amongst the nodes. (b) A network consisting of four nodes that is 3-complete, i.e., it features four sources that distribute tripartite entanglement. (c) The covariance matrix of the triangle network has a $3\times3$ block structure and consists of three terms, where 
$\qty(\square)_{i}$ denotes those blocks that are contributed by the source $i$.}
\label{fig:networks}
\end{figure}

{\it Quantum networks.---}
The simplest non-trivial network is the triangle network, where 
three nodes are mutually connected by three sources that prepare 
bipartite quantum states that are then subsequently shared with 
the nodes, see also Fig.~\ref{sfig:trianglenetwork}. More generally,
one has $M$ sources, labelled by  $m=1,2,\dots,M$ that independently 
produce quantum states $\varrho_m$, which are then distributed to $N$
nodes, labelled by $n=1,2,\dots,N$. For every source $m$ we denote by 
$\mathcal{C}_m$ the set of all connected nodes that have access to 
the state $\varrho_m$. The topology of the network captures the 
fact that not all vertices are connected to a single source, thus 
limiting the influence that each source can have on the different 
nodes. 

In the simplest case, at each node a measurement is performed that 
is described by a POVM  $\mathbf{A}^{(n)}=\{ A_{x}^{(n)}\}_{x}$. 
The observed probability distribution over the outcomes reads
\begin{equation}
\label{eq:networkdistribution}
p(x_1\dots x_N)=
\tr[(A_{x_1}^{(1)}\otimes\cdots\otimes A_{x_{N}}^{(N)})\varrho_1\otimes\cdots\otimes\varrho_M].
\end{equation}
The central question is whether a given probability distribution 
may originate from a network with a given topology. We note that 
the set of probability distributions that are compatible with a 
given network topology is non-convex and thus, in general, hard 
to characterize. One way to overcome this problem was put forward in Ref.~\citep{aberg2020}. The idea is to map the set of probability 
distributions compatible with the network to the space of covariance 
matrices, and then consider a convex relaxation of the problem. 

For this purpose, a so-called feature map is defined that maps the 
outcomes $x_n$ at each vertex $n$ to a vector $\mathbf{v}_{x_n}^{(n)}\in\mathcal{V}_n$, 
where the $\mathcal{V}_n$ are some orthogonal vector spaces. Combining 
all the feature maps, one obtains a random vector $\mathbf{v}$ with 
components $\mathbf{v}_{x_1\dots,x_N}=\mathbf{v}_{x_1}^{(1)}+\cdots+\mathbf{v}_{x_N}^{(N)}$. 
The covariance matrix is then defined as
\begin{equation}
\label{eq:CovMatrix}
\Gamma(\mathbf{v}) = 
E(\mathbf{v}\mathbf{v}^{\dagger})-E(\mathbf{v})E(\mathbf{v})^{\dagger}
\end{equation}
with $E(\mathbf{v}\mathbf{v}^{\dagger}) = \sum_{x_1,\dots,x_N}\mathbf{v}_{x_1,\dots,x_N}\mathbf{v}^{\dagger}_{x_1,\dots,x_N}P(x_1,\dots,x_N)$ and 
$E(\mathbf{v}) = \sum_{x_1,\dots,x_N}\mathbf{v}_{x_1,\dots,x_N}P(x_1,\dots,x_N)$. 
Due to the structure of $\mathbf{v}$, the covariance matrix has a natural
block structure: $\Gamma$ is an $N\times N$ block matrix with blocks $\Gamma_{\alpha\beta}$, and each block is a $r \times r$ matrix, with
$r$ being the dimension of $\mathcal{V}_n$. The standard covariance matrix 
formalism from mean values is a special instance of this notion, where one 
assigns to the outcomes $x_n$ just real numbers and hence takes the $\mathcal{V}_n$ 
to be one-dimensional. Here, however, we will assume that the feature 
map simply maps the outcome $x_n$ to $\ket{x_n}$, as for measurements 
with more than two outcomes the mean value contains less information in
comparison with the probability distribution.

{\it Covariance matrices and coherence.---}
The topology of the network imposes strong constraints on the structure 
of the covariance matrix. More precisely, the covariance matrix can be 
decomposed in a sum of positive matrices that have a certain block 
structure, corresponding to the sources \cite{aberg2020,kela2020}. 
The verification of this structure is then an instance of a semidefinite 
program (SDP)~\citep{boyd2004,Gaertner2012}. 
For simplicity, we will restrict our attention in the following to
$k$-complete networks. This means that all sources distribute their 
states to $k < N$ parties and all possible $k$-partite sources are 
being used, so we have $M={N \choose k}$ (see also Fig.~\ref{sfig:4completeNetwork}). Our results
can be extended to more complicated network topologies.

The criterion from Refs.~\cite{aberg2020,kela2020} states that one 
has to find a decomposition of $\Gamma (\mathbf{v})$ into blocks $Y_m$
according to
\begin{align}
\label{eq:SDPtest}
\text{find: } & Y_m \geq 0 \\
\text{subject to: } & Y_m = \Pi_m Y_m \Pi_m \label{eq:blockconstraint} 
\mbox{ and } \Gamma (\mathbf{v}) = \sum_{m=1}^{M} Y_m,
\end{align}
where $\Pi_m= \sum_{i\in \mathcal{C}_m} P_i$, with $P_n$ being the 
projector onto $\mathcal{V}_n$; so $\Pi_m$ is effectively a projector
onto all spaces affected by the source $m$. To give an example, we 
depict this decomposition for the case of the triangle network in Fig.~\ref{sfig:decomposition}. Note that the formulation in 
Eqs.~(\ref{eq:SDPtest}, \ref{eq:blockconstraint}) is different from
(but clearly equivalent to) the formulation in  
Ref.~\cite{aberg2020}. The advantage of our reformulation is that 
it allows to establish a link to the theory of quantum coherence.

When characterizing quantum coherence, one starts with a fixed basis
$\{  \ket{\phi_i} \}$ of the Hilbert space. The coherence of a quantum state 
is then given by the amount of off-diagonal elements of its density 
matrix, if expressed in this basis~\cite{Baumgratz2014, StreltsovReview}. A given pure state 
$\ket{\psi}$ is said to have coherence rank $k$, if it can be expressed
using $k$ elements of the basis $\{ \ket{\phi_i} \}$, and a mixed state
has coherence number $k$, if it can be written as a mixture of pure states
with coherence rank $k$~\cite{Killoran2016,StreltsovReview, Regula_2018, johnston2018, ringbauer2018}. This can be
extended to the notion of block coherence~\cite{kraft2019}. There, one
takes a set of orthogonal projectors $\qty{P_i}$ such that any vector 
$\ket{\psi}$ can be decomposed as $|\psi\rangle=\sum_{i}|\psi_{i}\rangle$, 
where $|\psi_{i}\rangle=P_{i}|\psi\rangle$. The vector $\ket{\psi}$ is said 
to have block coherence rank $k$ if exactly $k$ terms in the decomposition 
do not vanish. The convex hull of rank one operators $\ketbra{\psi}$ with 
block coherence rank $k$ we denote as $\mathcal{BC}_k$. Then, an operator 
$X$ has block coherence number $k+1$ if it is in $\mathcal{BC}_{k+1}$ 
but not in $\mathcal{BC}_{k}$. In general we have the inclusion
$\mathcal{BC}_1\subset\mathcal{BC}_2\subset\cdots\subset\mathcal{BC}_N$.
Note that the notions of coherence rank and coherence number
are well studied and several criteria and properties are known~\cite{StreltsovReview}.

Having this in mind, it is clear that Eqs.~(\ref{eq:SDPtest}, 
\ref{eq:blockconstraint}) are nothing but a reformulation of 
the notion of multisubspace coherence for the covariance matrix
and we arrive at the first main result of this paper:

\begin{observation}
If a covariance matrix $\Gamma(\mathbf{v})$ has block coherence number $k+1$, 
then it cannot have originated from a $k$-complete network.
\end{observation}

\begin{figure}[t!]
\subfloat[\label{sfig:triangleComp}]{
  \includegraphics[width=.47\linewidth]{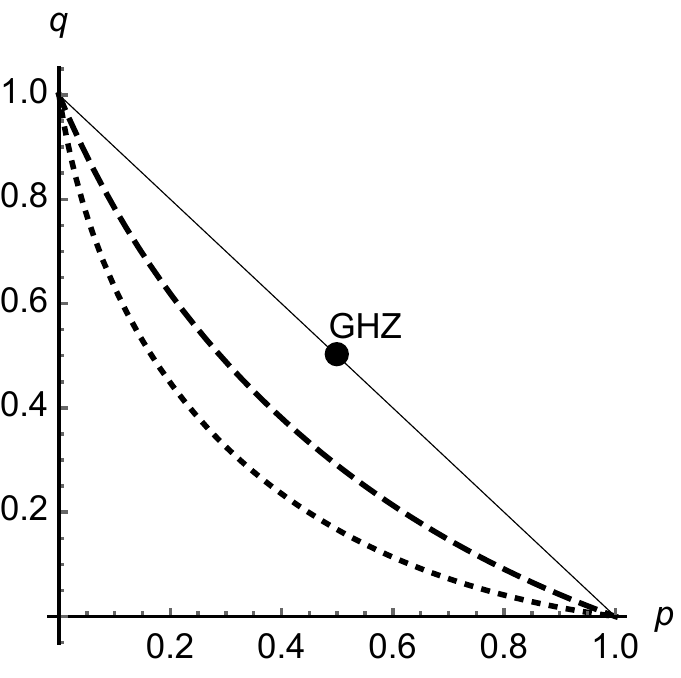}
}
\hfill
\subfloat[\label{sfig:fullN5}]{
  \includegraphics[width=.47\linewidth]{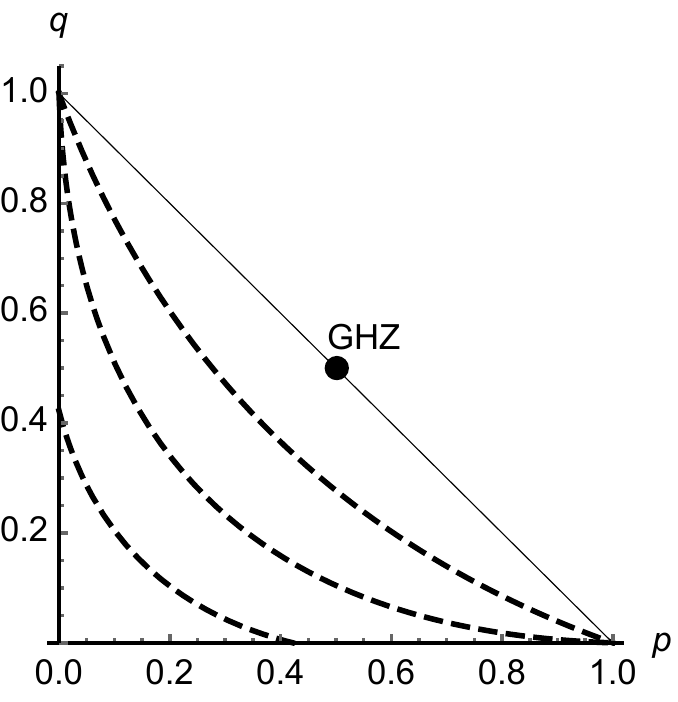}
}
\caption{(a) For the triangle network we compare the criterion in Eq.~\eqref{eq:aberg} (dotted line) to the monogamy criterion in Eq.~\eqref{eq:monogamy} (dashed line). (b) Results of the GHZ-type distribution in Eq.~\eqref{eq:GHZ2outcomes} using Eq.~\eqref{eq:resultGHZNk} for $N=5$ and $k=4,3,2$. Everything above the lines is detected to be incompatible with the respective network structure. $GHZ$ denotes the distribution that is obtained from measuring $\sigma_z$ on $\ket{GHZ}$.}
\label{fig:comparisionN}
\end{figure}

{\it Networks with dichotomic measurements.---}
For dichotomic measurements, that is, measurements with two 
outcomes, one can expect from our discussion after 
Eq.~(\ref{eq:CovMatrix}) that the covariance matrix can be 
simplified. Indeed, with our feature map the blocks of the 
covariance matrix are always of the form 
$(\Gamma_{\alpha \beta })_{ij}=(p_{ij}-q_i r_j)$, 
where $p_{ij}$ is a probability distribution, and 
$q_i=\sum_jp_{ij}$ and $r_j=\sum_ip_{ij}$ are its marginals. 
These blocks have 
vanishing row and column sums, so $(1, \dots, 1)^T$ is a 
(left and right) eigenvector to the eigenvalue zero. For 
the dichotomic case, the blocks are $2\times 2$ matrices, 
so only one nonzero eigenvalue remains, and we must have 
$\Gamma_{\alpha \beta } \propto (\id-\sigma_x)$. 
So we have:

\begin{observation}
\label{Obs:Factorize}
Consider a  network of $N$ vertices, where each node performs 
a dichotomic measurement. Then the covariance matrix $\Gamma(\mathbf{v})$ is 
of the form
\begin{equation}
\Gamma(\mathbf{v})=C\otimes (\id-\sigma_x),
\end{equation}
where $C$ is an $N\times N$ matrix.
\end{observation}

So, for evaluating the criterion for $k$-completeness in the case of dichotomic
measurements, one just has to check the $k$-level coherence of the matrix
$C$. While this is, in general, still hard, the solution can directly be written down for the simplest non-trivial case of $k=2$ \cite{ringbauer2018}. Namely, 
it is known that a matrix $X$ has coherence number less than or equal to two 
if and only if the so-called comparison matrix $M(X)$ defined by
\begin{equation}
(M[X])_{i j}=\left\{\begin{array}{ll}
\left|X_{i i}\right| & \text { if } i=j \\
-\left|X_{i j}\right| & \text { if } i \neq j
\end{array}\right.
\label{eq-comparisonmatrix}
\end{equation}
is positive semidefinite. Thus, we have:

\begin{observation}
\label{obs:comparisonMatrix}
If the comparison matrix $M(C)$ coming from the covariance matrix 
has a negative eigenvalue, then the observed probability 
distribution is incompatible with a network of bipartite 
sources.
\end{observation}

{\it Example of a GHZ-type distribution.---}
Consider the family of distributions that have previously been 
studied in Refs.~\cite{renou2019,aberg2020}
\begin{eqnarray}\label{eq:GHZ2outcomes}
    &P&(x_1,\dots ,x_N)=\notag\\
    &p&\delta_0^{(N)} + q\delta_1^{(N)} + (1-p-q)\frac{1-\delta_0^{(N)}-\delta_1^{(N)}}{2^N-2},
\end{eqnarray}
where $\delta_i^{(N)}=\prod_{j=1}^N\delta_{ix_j}$. For $p=q=\frac{1}{2}$ 
this corresponds to measuring locally $\sigma_z$ on an $N$-particle 
Greenberger-Horne-Zeilinger (GHZ) state $\ket{GHZ}= (\ket{00\dots0}+\ket{11\dots1})/\sqrt{2}$. The covariance 
matrix for this distribution reads
\begin{equation}\label{example}
C = \Delta \id + \chi\ketbra{\mathbf{1}},
\end{equation}
where $\Delta=2^{N-2}(1-p-q)/(2^{N}-2)$, 
$\chi=\tfrac{1}{4}[1-(p-q)^{2}]-\Delta$ 
and $|\mathbf{1}\rangle=\sum_{n=1}^{N}|n\rangle$. 
From Eq.~(\ref{eq-comparisonmatrix}) we can conclude that
$C$ has coherence number less or equal two if and only if the matrix $M(C)=(\Delta+2\chi)\id-\chi\ketbra{\mathbf{1}}$ is positive 
semidefinite. This matrix has eigenvalues $\lambda_1=\Delta+2\chi$ 
and $\lambda_2=\Delta-(N-2)\chi$.
It follows that $C$ is incompatible with a $2$-complete 
network if
\begin{equation}\label{eq:aberg}
q > p+\kappa-\sqrt{4\kappa p+(\kappa-1)^2},
\end{equation}
where $\kappa=[(N-1)2^{N-2}]/[(N-2)(2^{N-1}-1)]$.
This analytically recovers the numerical results from 
Ref.~\cite{aberg2020} and proves that the witness
conjectured in this reference is indeed optimal for arbitrary $N$.

{\it Multilevel coherence witnesses.---}
Due to the simple structure of the matrix $C$ in Eq.~(\ref{example})
we can completely characterize its multilevel coherence properties 
and so the underlying distributions according to their network 
topologies for arbitrary $N$. For this purpose we need the concept 
of coherence witnesses. Consider an arbitrary pure state 
$\ket{\psi}=\sum_{i=1}^M c_i\ket{i}$. A $(k+1)$-level coherence 
witness is given by~\cite{ringbauer2018}
\begin{equation}
W_k=\id-\frac{1}{\sum_{i=1}^k |c_i^\downarrow|^2}\ketbra{\psi},
\end{equation}
where $c_i^\downarrow$ denote the coefficients $c_i$ reordered decreasingly according to their absolute values. This means that $\tr[W_k \varrho]\geq 0$, if $\varrho$ has coherence number $k$ or less.
For the maximally coherent state $\ket{\psi^+}=(\sum_{i=1}^N\ket{i})/{\sqrt{N}}$ this witness is of the form $W_k=\id-\ketbra{\mathbf{1}}/k$. This witness can easily be proven to be optimal for the family of states
$\varrho(\mu)=\mu\ketbra{\psi^+}+(1-\mu)\frac{\id}{N}$, which is, 
up to normalisation and suitable choice of the parameter $\mu$, 
equivalent to $C$. Thus we obtain $\tr[W_kC]=(1-1/k)\Delta+(1-N/k)\chi$. 
From this, it directly follows that $C$ is incompatible with a $k$-complete 
network, if
\begin{equation}\label{eq:resultGHZNk}
q > p+\eta-\sqrt{4\eta p+(\eta-1)^2},
\end{equation}
with $\eta=(N-1)2^{N-2}/[(N-k)(2^{N-1}-1)]$. The results are shown 
for the case $N=5$ and $k=4,3,2$ in Fig.~\ref{sfig:fullN5}. Furthermore, we note that this technique can be applied to large networks where an approach based on SDPs would become infeasible, due to the rapidly growing number of terms in Eqs.~(\ref{eq:SDPtest}, \ref{eq:blockconstraint}), which grows as $N \choose k$.

{\it Networks beyond dichotomic measurements.---}
In the case of more than two outcomes per measurement, the 
block coherence number of the covariance matrix needs to 
be tested. For the case of networks involving only 
bipartite sources we have the following:

\begin{observation}
\label{Obs:BlockMinus}
Let $\Gamma(\mathbf{v}) \in \mathcal{BC}_2$ be a covariance matrix
with block coherence number two. Then, whenever the signs of some off-diagonal 
blocks are flipped such that the matrix remains symmetric, the resulting 
matrix will also remain positive semidefinite.
\end{observation}
To see this, note that any matrix with block coherence number two 
can be written as a convex combination of pure states with 
coherence rank two, i.e., $\ket{\psi}=P_i\ket{\psi}+P_j\ket{\psi}$. 
For any such state, adding a minus sign in the density operator 
corresponds to the transformation $P_i\ket{\psi}+P_j\ket{\psi}\mapsto P_i\ket{\psi}-P_j\ket{\psi}$, under which the density operator remains positive 
semidefinite. 

To demonstrate the power of this Observation, let us consider again 
the GHZ-type distribution, but now with three outcomes per measurement,
\begin{eqnarray}\label{eq:GHZ3outcomes}
    P(x_1,x_2,x_3)=&p&\delta_0^{(3)} + q\delta_1^{(3)} + r\delta_2^{(3)}\\ &+&(1-p-q-r)\frac{1-\delta_0^{(3)}-\delta_1^{(3)}-\delta_2^{(3)}}{3^3-3}.\notag
\end{eqnarray}
A straightforward calculation provides a regime where this is incompatible
with the triangle network, see Fig.~\ref{fig:threeoutcomes}.

\begin{figure}[t!]
\centering
  \includegraphics[width=.5\linewidth]{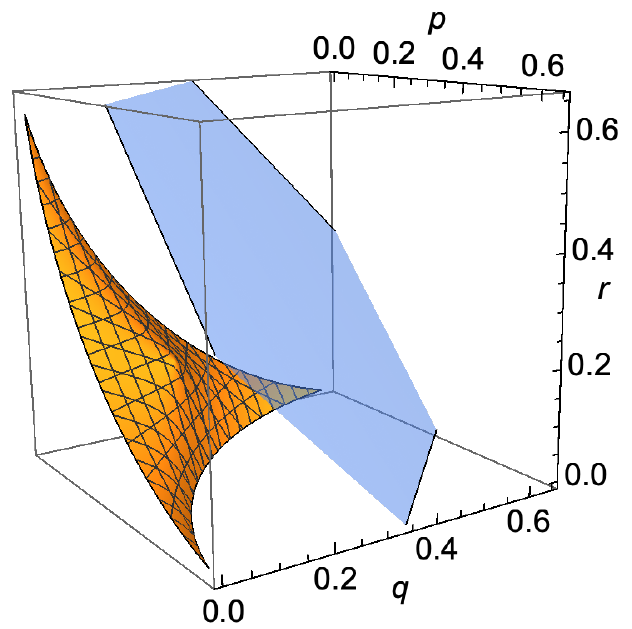}
\caption{Analysis of the GHZ-type distribution with three outcomes 
per measurement in Eq.~\eqref{eq:GHZ3outcomes} using Observation~\ref{Obs:BlockMinus}. Everything above the orange surface is detected to be incompatible 
with the triangle network. The blue surface represents the 
normalization constraint.}
\label{fig:threeoutcomes}
\end{figure}

{\it Characterizing networks with monogamy relations.---}
Another possibility to characterize networks is to evaluate monogamy 
relations for the coherence between different subspaces \cite{kraft2019}. 
The idea is that the amount of coherence that can be shared between 
one subspace and all other subspaces is limited if a certain block coherence 
number is imposed. To be more precise, for a trace one positive semidefinite 
block matrix $X=\qty[X_{\alpha\beta}]_{\alpha,\beta=0}^N$ with block coherence 
number $k$ it holds that
\begin{equation}
\sum_{\beta=1}^{N}\left\| X_{0 \beta}\right\|_{\rm tr}
\leq \sqrt{k-1}\sqrt{\tr[X_{00}](1-\tr[X_{00}])}.
\end{equation}
If we consider the normalized matrix matrix $\tilde{C}=C/\tr[C]$, 
evaluating such a monogamy relation provides a necessary 
criterion for $C$ to have coherence number $k$.
For the matrix in Eq.~(\ref{example}) this gives
\begin{equation}
\label{eq:monogamy}
\Delta-(\frac{\sqrt{N-1}}{\sqrt{k-1}}-1)\chi\geq 0.
\end{equation}
Hence, if this inequality is violated then the observed correlations 
are not compatible with a $k$-complete network. This is also shown in Fig.~\ref{sfig:triangleComp} for the triangle network. Although this 
test is in this case not as powerful as the analytical solution, it 
is easy to evaluate especially for large networks, since it requires 
only computing traces of smaller block matrices.

{\it Further remarks.---}
So far, we provided criteria to show that correlations
are incompatible with a $k$-complete network. It would be
interesting to derive also sufficient criteria for being
compatible with a given network structure. In the framework
of Ref.~\cite{aberg2020} this is not directly possible, 
as the criterion in Eqs.~(\ref{eq:SDPtest}, \ref{eq:blockconstraint})
is a convex relaxation of the original problem. Still, coherence 
theory allows to identify scenarios where the covariance matrix
can be certified to have a small block coherence number $k$, so the
covariance matrix approach must fail to prove incompatibility 
with a $k$-complete network. 

Here we can make two small observations in this direction: 
(i) The following results from Ref.~\cite{ringbauer2018} 
can directly applied to networks with dichotomic outcomes. 
Namely, if we have for the normalized matrix
$\tilde{C} \geq \frac{N-k}{N-1} \Lambda(\tilde{C})$, 
where $\Lambda$ is the fully decohering map, mapping any matrix to its 
diagonal part, then $\tilde{C}\in\mathcal{BC}_k$, implying that 
the test in Eqs.~(\ref{eq:SDPtest}, \ref{eq:blockconstraint}) for 
$(k+1)$-complete networks will fail. Furthermore we have that if 
$\tr[\tilde{C}^2]/\tr[\tilde{C}]^2\leq1/(N-1)$, 
then $\tilde{C}$ is two-level coherent.
(ii) In the general case, if $M_b(\Gamma) \geq 0$, where 
$M_b(\Gamma)$ is the block comparison matrix defined by
\begin{equation}
    (M_b[\Gamma])_{\alpha\beta}=
  \begin{cases}
    (\norm*{\Gamma_{\alpha\alpha}^{-1}})^{-1} & \text{for $\alpha=\beta$} \\
    -\norm{\Gamma_{\alpha\beta}} & \text{for $\alpha \neq \beta$,}
  \end{cases}
\end{equation}
with $\norm{X}$ denoting the largest singular value of the block $X$, 
then $\Gamma \in\mathcal{BC}_2$. A detailed discussion is given in the Appendix.

{\it Conclusion.---}
In this work we have established a connection between the theory 
of multilevel coherence and the characterization of quantum networks. 
To be precise, we showed that a recent approach based on covariance
matrices leads to a well studied problem in coherence theory; 
consequently, many results from the latter field can be transfered 
to the former. This provides a useful application of the resource 
theory of multilevel coherence outside of the usual realm of quantum 
states.

There are several interesting problems remaining for future work. 
First, it would be highly desirable to extend the covariance approach 
to the case where each node of the network can perform more than one
measurement. This will probably lead to significantly refined tests for network topologies. Second, it seems to be promising to study the coherence
in networks on the level of the resulting quantum state, and not the 
covariance matrix. This may shed light on the question which types of 
network correlations are useful for applications in quantum information
processing. 

\begin{acknowledgements}
This work was supported by the ERC (Consolidator Grant No. 683107/TempoQ), the DFG and the Austrian Science Fund (FWF): J 4258-N27.
\end{acknowledgements}

\appendix

\section{Appendix: Sufficient conditions for block coherence number two}

Let the block matrix $X=\qty[X_{\alpha\beta}]>0$, with 
$X_{\alpha\beta}\in\mathds{C}^{d\times d}$, be partitioned as follows
\begin{equation}\label{eq:paritioning}
X=\left[\begin{array}{cccc}
X_{11} & X_{12} & \cdots & X_{1 K} \\
X_{21} & X_{22} & \cdots & X_{2 K} \\
\vdots & \vdots & \ddots & \vdots \\
X_{K 1} & X_{K 2} & \cdots & X_{KK}
\end{array}\right].
\end{equation}

\begin{definition}[from Ref.~\cite{feingold1962}]
Let $X$ be partitioned as in Eq.~\eqref{eq:paritioning}. 
If the matrices $X_{\alpha \alpha}$ on the diagonal 
are non-singular, and if
\begin{equation}\label{eq:blockdiagdominant}
    (\norm*{X_{\alpha \alpha}^{-1}})^{-1}
    \geq\sum_{\substack{\beta=1\\\beta\neq \alpha}}^K \norm{X_{\alpha\beta}},
\end{equation}
then $X$ is called \emph{block diagonally dominant}. Here, $\norm{Y}$ denotes 
the largest singular value, so for the positive  $X_{\alpha \alpha}$ the expression 
$\norm{X_{\alpha \alpha}^{-1}}^{-1}$ is the smallest eigenvalue of 
$X_{\alpha \alpha}$.
\end{definition}

\begin{observation}\label{obs:coherenceNumberTwo}
If $X$ is positive and block diagonally dominant, then
the block coherence number is smaller or equal to two, $bcn(X)\leq 2$.
\end{observation}
\begin{proof}
Suppose $X$ satisfies the hypothesis. Define $2\times 2$ block matrices
\begin{equation}
    G^{\alpha\beta}=\mqty[\abs{X_{\alpha\beta}} && X_{\alpha\beta} \\ X_{\alpha\beta}^{\dagger} && \abs{X_{\alpha\beta}^{\dagger}}],
\end{equation}
where $\abs{X_{\alpha\beta}}=\sqrt{X_{\alpha\beta}^{\dagger}X_{\alpha\beta}}$ and the support of $G^{\alpha\beta}$ is the subspace $\alpha,\beta$. Clearly, the $G^{\alpha\beta}$ are positive semidefinite and have block coherence number two. Next, consider the matrix $D=X - \sum_{\alpha=1}^K\sum_{\beta> \alpha} G^{\alpha\beta}$. Since $X>0$ it is also hermitian, and thus, $X_{\beta\alpha}=X_{\alpha\beta}^{\dagger}$, precisely as for $G^{\alpha\beta}$. From this we can conclude that the off-diagonal blocks of $D$ vanish and the diagonal blocks are given by $D_{\alpha \alpha }=X_{\alpha \alpha }-\sum_{\substack{\beta=1, \beta\neq \alpha}}^K\abs{X_{\alpha \beta}}$. Furthermore, observe that $\lambda_{\text{min}}(X_{\alpha \alpha}) \geq \sum_{\substack{\beta=1, \beta\neq \alpha}}^K\lambda_{\text{max}}(X_{\alpha \beta})\geq\lambda_{\text{max}}(\sum_{
\substack{\beta=1, \beta\neq \alpha}}^K X_{\alpha \beta})$, where the first inequality is due to Eq.~\eqref{eq:blockdiagdominant} and the second inequality is straightforward. This proves that, besides being block diagonal, $D$ is also positive semidefinite. Thus $X$ can be written as a positive sum of a block incoherent matrix $D$ and matrices $G^{\alpha\beta}$ of block coherence 
number two, from which the statement follows.
\end{proof}

The next concept that is needed is the so-called comparison matrix, 
which is defined as follows.

\begin{definition}[from Ref.~\cite{polman1987}]
Let $X$ be partitioned as in Eq.~\eqref{eq:paritioning} and $X_{\alpha \alpha}$ non-singular. Then the block comparison matrix $M_b[X]$ is defined by
\begin{equation}
    (M_b[X])_{\alpha \beta}=
  \begin{cases}
    (\norm*{X_{\alpha \alpha }^{-1}})^{-1} & \text{for $\alpha=\beta$} \\
    -\norm{X_{\alpha \beta}} & \text{for $\alpha \neq \beta $.}
  \end{cases}
\end{equation}
\end{definition}
From this definition it is evident that if the comparison matrix $M_b[X]$ exists and is (strictly) diagonally dominant, then $X$ itself is (strictly) block diagonally dominant.

\begin{definition}[M-matrix]\label{def:Mmatrix}
Let the matrix $A=(a_{ij})$ be a real matrix such that $a_{ij}\leq0$ for $i\neq j$. Then $A$ is called a nonsingular M-matrix if and only if every real eigenvalue of $A$ is positive.
\end{definition}

\begin{definition}[Def. 3.2. in Ref.~\cite{polman1987}]\label{def:blockHmatrix}
If there exist nonsingular block diagonal matrices $D$ and $E$ such that $M_b[DXE]$ is a nonsingular M-matrix, then $X$ is said to be a nonsingular block H-matrix.
\end{definition}

\begin{lemma}[Lemma 4. in Ref.~\cite{polman1987}]\label{lem:blockdiagonal}
If $X$ is a nonsingular block H-matrix then there exist nonsingular block diagonal matrices $D$ and $E$ such that $DXE$ is strictly block diagonally dominant.
\end{lemma}

\begin{theorem}
Let $X$ be partitioned as in Eq.~\eqref{eq:paritioning} and positive semidefinite (but not necessarily strictly positive). If $M_b(X)\geq0$, then $X$ has $bcn(X)\leq2$.
\end{theorem}
\begin{proof}
The proof follows the idea of Ref.~\cite{ringbauer2018}. First, define the operator $X_{\epsilon}=X+\epsilon \id$, for $\epsilon\geq0$. Then, for $\epsilon>0$ we have that $M_b[X_\epsilon]=M[X]+\epsilon \id > 0$. Evidently, since $M_b[X_\epsilon]$ is a real matrix with non-positive off-diagonal entries and furthermore has only strictly positive eigenvalues it is a nonsingular M-matrix, according to Def.~\ref{def:Mmatrix}. Then, according to Def.~\ref{def:blockHmatrix} $X_\epsilon$ is a nonsingular block H-matrix. From the proof of Lemma~\ref{lem:blockdiagonal} in Ref.~\cite{polman1987} we can conclude that there exists a block diagonal matrix $D>0$ such that $D_{\epsilon}X_{\epsilon}D_{\epsilon}$ is strictly block diagonally dominant. Then it follows from Observation~\ref{obs:coherenceNumberTwo} that strictly block diagonally dominant matrices can have at most block coherence number two. We find that $bcn(X_\epsilon)=bcn(D_{\epsilon}X_{\epsilon}D_{\epsilon})\leq2$, and since the block coherence number is 
lower semi-continuous we have $bcn(X)=bcn(\lim_{\epsilon\rightarrow 0^+} X_\epsilon)\leq \lim_{\epsilon\rightarrow 0^+} bcn(X_\epsilon)\leq 2$.
\end{proof}

\end{document}